\documentclass[a4paper]{article}
\usepackage{amssymb,amsmath,amsthm,url}
\usepackage{graphics,palatino}
\usepackage{tikz,pgffor}
\DeclareMathOperator{\K}{\mathit{K}\,}
\newcommand{\0}{\mathbf{0'}}
\DeclareMathOperator{\KH}{\mathit{K}^\0\,}

\newcommand{\logg}{\log^{(2)}}

\newcommand{\Plsc}{\mathcal{P}^\uparrow}

\theoremstyle{plain}

\newtheorem{theorem}{Theorem}
\newtheorem*{theorem*}{Theorem}
\newtheorem{proposition}[theorem]{Proposition}
\newtheorem{lemma}[theorem]{Lemma}
\newtheorem{corollary}[theorem]{Corollary}

\theoremstyle{remark}

\newenvironment{keywords}{\noindent\sc Keywords:\rm}{\mbox{}}

\begin{document}

\title{Upper-semicomputable sumtests for lower semicomputable semimeasures}
\author{Bruno Bauwens\footnote{%
    LORIA, Universit\'e de Lorraine. 
    This research was supported by a specialisation grant form 
    the Flemish agency for innovation through science and technology (IWT-Vlaanderen).
    I'm grateful for social and practical support from the SYSTMeS and neurology 
    departments of Ghent University by Georges Otte, Luc Boullart, Bart Wyns and Patrick Santens 
    in the period 2008--2010.
    The research was finished with support from Universit\'e de Lorraine in November 2013.
    } \footnote{
    The results have appeared in my Phd thesis~\cite{BauwensPhd} (with cumbersome proofs).
    I thank Sebastiaan Terwijn for sharing some proof attempts for results below and 
    Alexander (Sasha) Shen for the simplified proof of Theorem~\ref{th:u_hierarchy}, and 
    so much interesting discussion about Kolmogorov complexity.
    }
}

\date{}

\maketitle  

\begin{abstract}
  A sumtest for a discrete semimeasure $P$ is a function $f$ mapping 
  bitstrings to non-negative rational numbers such that
  \[
    \sum P(x)f(x) \le 1 \,.
  \]
  Sumtests are the discrete analogue of Martin-L\"of tests.
%  For pairs of strings, algorithmic mutual information $I(x;y)$ defines a sumtest relative 
%  to a product of universal lowersemicomputable semimeasures $m(x)m(y)$. 
%  In [B.Bauwens and S.Terwijn, Theory of Computing Systems 48.2 (2011): 247-268], it is shown that
%  for this product semimeasure, no unbounded lower semicomputable sumtests exist, 
%  but a hierarchy of upper semicomputable tests exist approaching $I(x;y)$.
%
  The behavior of sumtests for computable $P$ seems well understood,
  but for some applications lower semicomputable $P$ seem more appropriate. 
  In the case of tests for independence, 
  it is natural to consider upper semicomputable tests (see
  [B.Bauwens and S.Terwijn, Theory of Computing Systems 48.2 (2011): 247-268]).

  In this paper, we characterize upper semicomputable sumtests relative to any lower semicomputable semimeasures 
  using Kolmogorov complexity.
  It is studied to what extend such tests are pathological: can
  upper semicomputable sumtests for $m(x)$ be large?
  It is shown that the logarithm of such tests does not exceed $\log |x| + O(\logg |x|)$ (where $|x|$ denotes  
  the length of $x$ and $\logg = \log\log$) and that this bound is tight, i.e. 
  there is a test whose logarithm exceeds $\log |x| - O(\logg |x|$) infinitely often.
  Finally, it is shown that for each such test $e$ 
  the mutual information of a string with the Halting problem is at least $\log e(x)-O(1)$;
  thus $e$ can only be large for ``exotic'' strings.
  %Using techniques above we also show that the normalized information distance is not $n$-difference 
  %computable for any $n$.
\end{abstract}

\begin{keywords}  Kolmogorov complexity -- universal semimeasure -- sumtest 
\end{keywords}

\section{Introduction} 
\label{sec:introduction}

A (discrete)  {\em semimeasure} $P$ is a function from strings to non-negative rational numbers such that 
$\sum P(x) \le 1$. %, the class of semimeasures is denoted by $\mathcal P$.
A {\em sumtest} $f$ for a semimeasure $P$ is a function from bitstrings 
to non-negative reals such that 
\[
 \sum_x P(x)f(x) \le 1\,.
\]
Sumtests provide a rough model for statistical significance testing 
for a hypothesis about the generation of data~\cite{LiVitanyi, statCourse, Young}. 
This hypothesis should be sufficiently specific so that it determines 
a unique semimeasure $P$ describing 
the generation process of observable data $x$ in a statistical experiment.
The value $f(x)$ plays the role of significance 
in a statistical test. % and $x$ corresponds to the observed data. 
The condition implies that the $P$-probability of observing $x$ 
such that $f(x) \ge k$ is bounded by $1/k$.
If $f(x)$ is high, it is concluded that either a rare event has occurred or the 
hypothesis is not consistent with the generation process of the data.

For many hypotheses, such as the hypotheses that two observables are independent,
not enough information is available to infer a unique probability distribution.
In such a case, one might consider a set of semimeasures that are consistent with the hypothesis. 
We say that $m^{\mathcal H}$ is {\em universal} in a set of semimeasures if $m^{\mathcal H}$ 
is in the set and if for 
for each $P$ in the set there is a constant $c$ such that $P \le c\cdot m^{\mathcal H}$. 
(If $P \le c\cdot Q$ for some $c$, we say that $Q$  {\em dominates} $P$.)
%Under some 
%In this work we add the hypothesis that semimeasures are lower semicomputable.
%\footnote{Computable (semi)measures do not imply a nice theory, upper semicomputable measures
%  are equivalent to computable ones and upper semicomputable semimeasures can be argued to be less
%  natural in the context of learning theory.}
It might happen that the class of lower semicomputable $P$ consistent with a 
hypothesis has a universal element $m^{\mathcal H}(x)$ and that each such element satisfies
\[
P(x) \le O\left(2^{-K(P)}\right)m^{\mathcal H}(x) \,,
\]
here, $\K(P)$ is the minimal length of a program that enumerates $P$ from below.
This happens for the hypothesis of independence of two strings, or directed influences in time
series (see~\cite[Proposition 2.2.6]{BauwensPhd}).
In this case, if observed data results in a high value of a sumtest for $m^{\mathcal H}$, 
it can be concluded that either: 
\begin{itemize}
 \item a rare event has occurred,
 \item the hypothesis is not consistent with the 
  generation of the data,
 \item the data was generated by a process that is only consistent with 
   semimeasures of high Kolmogorov complexity.
\end{itemize}

Unfortunately, the use of approximations of universal tests 
seems not to be practical, and this interpretation
is rather philosophic.
%(This is a rather philosophic observation, it is very hard to approximate $m^{\mathcal H}$ in
%practice.)
However, one can raise the question whether sequences of improving 
tests (for example general tests for independence) 
reported in literature tend to approach some ideal limit. 
This motivates the question whether there exist
large sumtests in some computability classes, and 
whether they have universal elements.
%More generally, we say that a function $f$  {\em dominates} a function $g$ 
%if $g \le cf$ for some non-negative $c$.
%A function $f$ is {\em universal} in a set $S$ of functions if $f \in S$ 
%and $f$ dominates each other test in the set. 

Let $\Plsc$ be the class of lower semicomputable semimeasures.
In algorithmic information theory, it is well known that this class has
a universal element~$m$ \cite{LiVitanyi}, and $m$ can be characterized in terms of 
prefix Kolmogorov complexity: $\log m(x) = \K(x) + O(1)$ for all~$x$ (see~\cite{LiVitanyi} for more
background on Kolmogorov complexity).
For all computable $P$ it is also known that:
\begin{itemize}
  \item 
    no universal element exists in the class of computable tests,
  \item 
    no upper semicomputable test exists that dominates all lower semicomputable ones, 
  \item 
    a universal sumtest in the class of lower semicomputable tests is 
    \[
      m(x)/P(x) \,.
    \]
\end{itemize}
In \cite{BauwensTerwijn}, these statements are studied for
lower semicomputable $P$. In this case all of the results above 
become false: some semimeasures have a computable universal test 
(for example $f(x) = 1$ is a universal test for~$m$ even among the lower semicomputable tests), 
upper semicomputable tests can exceed lower semicomputable ones, 
(every $P \in \Plsc$ has an unbounded upper semicomputable test~\cite[Proposition 5.1]{BauwensTerwijn}, thus also~$m$),
and for some $P$ no universal lower semicomputable test exists~\cite[Proposition
4.4]{BauwensTerwijn}.

Consider the hypothesis of independence for bivariate
semimeasures. The corresponding class of semimeasures for this hypothesis is $P(x)Q(y)$ 
where $P$ and $Q$ are univariate semimeasures. The subset of 
lower semicomputable semimeasures has a universal element given by 
$m(x)m(y)$. A sumtest relative to this semimeasure can be called an independence test.
%to the most general hypothesis of independence \cite{Levin, AIT}.
It is not hard to show that any lower semicomputable independence tests is bounded 
by a constant\footnote{
   Suppose that a sumtest $f$ for $m(x)m(y)$ is unbounded. 
   For each $k$ one can search a pair $(x,y)$ such that $f(x,y) \ge 2^k$. For the first such pair
   that appears we have $\log m(x) = \K(x)  \le K(x,y)  \le K(k) \le O(\log k)$ up to $O(1)$ terms, 
   and similar for $m(y)$. This implies $f(x,y) m(x)m(y) \ge 2^{k - O(\log k)}$, 
   which contradicts the condition in the definition of sumtest.
   }.
In~\cite{BauwensTerwijn} it is shown that upper semicomputable independence tests 
exist whose logarithm equal $n + O(\log n)$ for all $n$ and pairs $(x,y)$ of strings of length $n$.
Moreover, there exist a generic upper semicomputable test $u_h$ defined for all
computable $h$, and this test is increasing in $h$: if $h \le g$
then $u_h \le u_g$.
The tests obtained in this way somehow ``cover'' all tests:
every upper semicomputable test is dominated by $u_h$ for some computable $h$.
Moreover, $u_h$ has a characterization in terms of (time bounded) Kolmogorov complexity, and 
for fixed $(x,y)$ and increasing $h$, the value $\log u_h(x,y)$ approaches algorithmic mutual
information $I(x;y) = \K(x) + \K(y) - \K(x,y)$.
Unfortunately, no universal upper semicomputable sumtest exist.

The first result of the paper is to define such a generic test $u_h$ for 
all lower semicomputable semimeasures in terms of Kolmogorov complexity. This generalizes 
the result mentioned above. The remaining goal of the paper is to
investigate whether upper semicomputable tests are pathological, 
in particular can such tests for the universal semimeasure $m$ be large? 
This would imply that these tests identify ``structure'' which can not be identified by 
compression algorithms. 

We show that there are sumtests for $m$
for which the logarithm exceeds $\log |x| - O(\log\log |x|)$, 
but for no test its logarithm exceeds $\log |x| + O(\log |x|)$.
This is small compared to the logarithm of the independence tests discussed above, 
which equal $n + O(\log n)$ for some $(x,y)$ of length $n$. 
We show that the logarithm of each test is bounded by the mutual information with the Halting problem $\0$ 
given by $\K(x) - \KH(x)$, up to an additive constant (that depends on the test).
Hence, strings $x$ for which such a test is large, 
are unlikely to be produced in a statistical experiment.
We also show that there is no universal upper semicomputable test for~$m$.

%Our results imply that
%sumtests for the universal semimeasure can only be high
%for ``gradual compressible'' strings.  
%Gradual compressibility, means that if more computation 
%time is available, more structure is found in the data.
%Our construction of gradual compressible strings allows to 
%solve another open problem from~\cite{compNID}: is the normalized 
%information distance $n$-difference computable for some $n$.
%We show this is not the case.
%This often appears in the process of learning 
%incremental abstraction levels. 
%Such strings are ``complex'', this follows by 
%Proposition \ref{prop:sdoVsMms} by observing that the 
%value of a sumtest for a universal semimeasure of a finite 
%binary string $x$, is dominated by the length of a minimal 
%sufficient statistic of $x$ within a logarithmic term.

\section{A generic upper semicomputable sumtest}
\label{sec:hierarchy}

Let $P$ be a lower semicomputable semimeasure. For every computable two argument function $h$, 
we define an upper semicomputable sumtest $u_h$ for $P$. For increasing $h$ these sumtests are increasing
and we show that for each upper semicomputable sumtest $f$ there exists a computable $h$ such 
that $u_h$ exceeds $f$ up to a multiplicative constant.

In our construction we use $m(\cdot|\cdot)$, 
which is a bivariate function that is universal for all lower semicomputable conditional
semimeasures. Let $m_t(\cdot|\cdot)$ and $P_t(\cdot)$ represent approximations of $m(\cdot|\cdot)$ and
$P(\cdot)$ from below.  

For any computable $h$, let
$$
u_{h}(x) = \inf_s \left\{\frac{m_{h(x,s)}(x|s)}{P_s(x)}\right\}.
$$
%And let $u_{h,t}(x)$ denote this minimum over all $s \le t$.

\begin{theorem} \label{th:u_hierarchy}
  $u_{h}$ is an upper semicomputable sumtest for $P$. For any upper semicomputable sumtest $e$ for~$P$, 
  there exist a $c$ and a computable $h$ such that $u_{h} \le c \cdot e$.
\end{theorem}

 {\em Remark.} Let $t$ be a number. We can define a generic test $u_h$ using 
time bounded Kolmogorov complexity, given by 
\[
K^t\,(x) = \min \left\{ |p|:U(p) \text{ outputs $x$ and halts in at most $t$ steps} \right\}. 
\]
Indeed, we simply fix the approximation $m_t(x|s)$ to be $2^{-K^t(x|s)}$, 
and by the conditional coding theorem, this defines a universal conditional
semimeasure $m(\cdot|\cdot)$. 

\begin{proof}\hspace{-1mm}\footnote{
    This simplified proof was suggested by Alexander Shen.
    }\hspace{1mm}
  Clearly, $u_{h}$ is upper semicomputable, so we need to show it is a sumtest.
  For each fixed $s$, the function $m_{h(x,s)}(x|s)/P_s(x)$ is a sumtest for $P_s$, 
  and $u_h(x)$ is not larger; thus for all $s$:
  \[
    \sum_x u_h(x)P_s(x) \le 1\,. 
  \]
  This implies that the relation also holds in the limit, thus $u_h$ is a sumtest for~$P$.

  For the second claim of the lemma, note that we can 
  choose an approximation $e_s$ (from above) of $e$  such that 
  $e_s$ has computably bounded support and $\sum_x P_s(x)e_s(x) \le 1$ for all~$s$. 
  By universality of $m(\cdot|\cdot)$, this implies that there exist a $c$ such
  that $P_s(x)e_s(x) \le c\cdot m(x|s)$, and this $c$ does not depend on $x$ and~$s$.
  We can wait for the first stage in the approximation of $m(x|s)$ for which this equation becomes
  true and let this stage define the function $h(x,s)$.
  This implies $e_s(x) \le c\cdot m_{h(x,s)}(x|s)/P_s(x)$ for all $x$ and $s$, thus $e(x) \le c\cdot u_h(x)$. 
  $h(x,s)$ is defined for all $x$ and $s$, thus $h$ is computable.
\end{proof}

Let $\KH(x)$ be the Kolmogorov complexity of $x$ on a machine with an oracle for to the Halting problem 
that is optimal.
\begin{corollary}\label{cor:high_u_muchHaltingInformation}
  If $e$ is an upper semicomputable sumtest for a universal semimeasure, then
  $\log e(x) \le \K(x) - \KH(x) + O(1)$. 
  (The constants implicit in the $O(\cdot)$ notation depend on $e$.)
\end{corollary}

\begin{proof}
  %\hspace{-1mm} \footnote{ 
  %  A self contained proof exists where the result is shown up to additive terms $O(\log |x|)$.
  %  Choose $s$ in the definition of $u_h$ to 
  %  be the largest computation time of a program of length $u_h(x)$ minus small terms. 
  %  $u_h \le m_t(x|t)/m(x)$ and 
  %  This shows that $m_t(x)$ must converge slowly, and this can only happen for $x$ with 
  %  large mutual information with the Halting problem. The proof presented in the text, 
  %  gives a more precise result.
  %  }\hspace{1mm}  
  It suffices to show the corollary for $e = u_h$.
  Let $m^\0(x) = 2^{-\KH(x)}$.
  We use~\cite[Theorem 2.1]{limitComplexitiesRevisited}, which states
  \[
  m^\0(x) = \Theta \left( \liminf_t m(x|t) \right) \,.
    \]
  By definition of $u_h$:
  \[
   u_h \le \liminf_t \frac{m_{h(x,t)}(x|t)}{m_t(x)} \,.
  \]
  For all but finitely many $t$, the denominator exceeds $m(x)/2$, and 
  by the theorem mentioned above, 
  $\liminf m(x|t)$ in the numerator is $O(m^\0(x))$.
  By choice of $m^\0(x)$, it equals $2^{-\KH(x)}$. Thus
  \[
  \le \liminf_t \frac{m(x|t)}{m(x)/2} \le O\left( \frac{m^\0(x)}{m(x)}\right) = O\left(
  \frac{2^{-\KH(x)}}{2^{-\K(x)}} \right) \,. 
  \qedhere
  \]
\end{proof}

\section{Upper bound for upper semicomputable tests for~$m$}

\begin{theorem}\label{th:upperbound}
  If $e$ is an upper semicomputable sumtest for $m$, then 
  \[
    e(x) \le O\left( |x|(\log |x|)^2 \right) \,.
    \]
\end{theorem}

\begin{proof}
 By  Theorem~\ref{th:u_hierarchy}, it suffices to show the theorem for $u_h$ for all computable $h$.  
 Note that $\sum_x 2^{-2|x|} = \sum_n 2^{-n} = 1$, thus for each universal semimeasure~$m$ 
 there exist $c>0$ such that $m(x) > c2^{-2|x|}$.
 Assume that $m_t(\cdot)$ is an approximation from below of such $m(\cdot)$ such that 
 $m_1(x) \ge c2^{-2|x|}$.  

 The idea is as follows. We consider some times $t_1, t_2, \dots$ 
 on which $m(t)$ is large. 
 (For any number $n$ let $m(n)$ be 
 the universal semimeasure for the string containing $n$ zeros.)
 This implies that 
 $m_{h(x,t_i)}(x|t_i)$ is not much above $m_{t_{i+1}}(x)$ if $t_{i+1}$ 
 is sufficiently above $t_i$. From the definition of $u_h(x)$ with $t = t_i$ 
 it follows that if $u_h(x)$ is large, then also $m_{t_{i+1}}(x)/m_{t_i}(x)$ must be large. 
 On the other hand, $m(x)/m_1(x)$ is bounded, thus the first 
 ratio can only be large for few $t_i$. On the other hand, 
 our construction implies that for large $u_h(x)$ the first ratio must be large 
 for many $t_i$.

 We show the following claim

 \textit{ 
 For all computable $h$ there exist a series of numbers $t_1, t_2, t_3, \dots$ and a constant $c > 0$ 
 such that for all $i \ge |x|$ either 
 \[
 \frac{m_{t_{i+1}}(x) }{m_{t_i}(x)} \ge 2
    \quad\quad \text{or} \quad\quad
 u_h(x) < 2ci(\log i)^2 \,.
 \]
 }
 Let us first show how this implies the theorem. 
 The definition of a semimeasure implies $m(x) \le 1$, thus 
 $
 \frac{m(x)}{m_1(x)} \le O\left(2^{2|x|}\right)
 $ by assumption on $m_1(\cdot)$, and hence
 \[
  \frac{m_{t_2}(x)}{m_{t_1}(x)} \frac{m_{t_3}(x)}{m_{t_2}(x)}\dots
 \frac{m_{t_{4|x|}}(x)}{m_{t_{4|x|-1}}(x)}  \le \frac{m(x)}{m_1(x)} \le O\left(2^{2|x|}\right)\,.
   \] 
   For large $x$, at most $2|x|+O(1) < 3|x|$ elements in the series $t_{|x|}, t_{|x| + 1},
   \dots, t_{4|x|-1}$ can satisfy the left condition. Thus, some element does not satisfy the condition 
   and hence
 \[
   u_h(x) < 2c\left(4|x|\right)\left(\log (4|x|)\right)^2\,.
   \] 
 This implies the theorem.

 %thus $\Omega\left(k/(\log k)^2 \right) \le |x| + O(1)$, i.e. $k \le O\left( |x|(\log |x|)^3 \right)$.

 \medskip
 We now construct a sequence $t_1, t_2, t_3,\dots$ satisfying the conditions of the claim.
 This construction  depends on a parameter $c$, which will be chosen later. Let $t_1 = 1$.  
 For $i \ge 1$, $t_{i+1}$ is given by the first stage in the approximation of~$m(\cdot)$ such that
 \begin{equation}\label{eq:construction_t}
   \frac{m_{h(x,t_i)}(x|t_i)}{ i(\log i)^2} \le c\cdot m(x) \,,
 \end{equation}
 for all $x$ of length at most~$i$.

 We first argue why for appropriate $c$, such a stage $t_{i+1}$ exist, i.e. why 
  \eqref{eq:construction_t} holds for all $x$.
 Note that the sequence is recursively enumerated uniformly in $c$,
 thus $m(c)/(i(\log i)^2) \le O(m(i)m(c)) \le O(m(t_i))$. 
 On the other side 
 \[
 m_{h(x,t_i)}(x|t_i)m(t_i) \le O(m(x))\,,
 \]
 thus for some $c'$ independent of $c$, $i$ and $x$:
 \[
 \frac{m_{h(x,t_i)}(x|t_i) \,m(c)}{ i(\log i)^2} \le c'\cdot m(x) \,.
 \]
 and~\eqref{eq:construction_t} is satisfied if $c \ge c'/m(c)$. 
 This relation holds if we choose $c$ to be a large power of two 
 (indeed $m(2^l) \le \alpha/l^2$ for some $\alpha>0$, thus choose $l$ such that $2^l \ge c'l^2/\alpha$).

 It remains to show the claim. 
 Assume that the right condition is not satisfied and choose $t = t_i$ in the definition of $u_h$: 
 \begin{eqnarray*}
  2ci(\log i)^2 \le u_h(x) &\le& \frac{m_{h(x,t_i)}(x|t_i)}{m_{t_{i}}(x)} 
  = \frac{m_{h(x,t_i)}(x|t_i)}{m_{t_{i+1}}(x)}\frac{m_{t_{i+1}}(x)}{m_{t_{i}}(x)} \\
 & \le &  ci (\log i)^2 \frac{m_{t_{i+1}}(x)}{m_{t_{i}}(x)}\,.
 \end{eqnarray*}
 This implies the left condition.
\end{proof}

\section{Construction of large upper semicomputable tests}

For each pair $(f,g)$ of computable functions, an upper semicomputable function 
$e_{f,g}$ is constructed. Afterwards, it is shown that $\sum_x m(x)e_{f,g}(x) \le O(1)$  for appropriate~$f$.
Finally, we construct $g$ such that $e_{f,g}$ equals $|x|/(\log |x|)^5$ for infinitely many $x$.
Before constructing $e_{f,g}$, we show a technical lemma.

\begin{lemma}\label{lem:t_k}
 There exists a sequence of numbers $t_1,t_2,\dots$ such that for all $k$
 \[
 \sum_x \left\{m(x) : \frac{m(x)}{m_{t_k}(x)} \ge 2 \right\} \le 2^{-k+1} 
 \]
 and such that for some computable function $f$ and large $k$
 \[
   m_{f(t_k)}(x) \le 2^{-k}/k^2 \,.
 \]
\end{lemma}

\begin{proof}
  The proof of the lemma is closely related to the proof that strings with high Kolmogorov complexity
  of Kolmogorov complexity given the string are rare~\cite{GacsNotes} (see~\cite{msoph} for more on
  the technique). 
  The construction of $t_k$ uses the function $k_t$, which is in turn
  defined using an approximation from below for the famous number $\Omega$ (see~\cite{Chaitin75}):
  \[
  \Omega = \sum_x m(x) \quad \quad \text{and} \quad \quad 
  \Omega_t = \sum_{|x| \le t} m_t(x)\,.
  \]
  For each $t$ let $k_t$ be the position of the leftmost bit of $\Omega_t$ (in binary) that 
  differs from $\Omega_{t-1}$. Note that $k_t$ tends to infinity for increasing $t$.
  Let 
  \[
  t_k = \max\left\{ t : k_t \le k \right\}\,, 
  \]
  i.e. the largest $t$ for which there is a change in the first $k$ bits of $\Omega_t$.
  Clearly, 
  \[
  \sum_x m(x) - \sum_{|x|\le t_k} m_{t_k}(x) = \Omega - \Omega_{t_k} \le 2^{-k}\,.
  \]
  and this inequality implies the first inequality of the lemma.
  \medskip

  For the second, we show that there exist a computable $f$ such that for all $t$:
  \[
  m_{f(t)}(t) \le O\left( 2^{-k_t}/k_t^2\right) \,.
  \]
  Indeed, given the first $k_t$ bits $y$ of $\Omega_t$ we can compute $t$ 
  (by waiting until the first stage $s$ such that $y$ is a prefix of $\Omega_s$).
  This implies $m(y) \le O(m(t))$. 
 Note that $2^{-|z|}/|z|^2 \le o(m(z))$ for all $z$, 
 thus 
 \[ 
   2^{-k_t}/k_t^2 \le m(t)
 \]
 for large $t$.
 $k_t$ is computable from $t$ and we can wait until the current approximation of $m(t)$ 
 is large enough to satisfy the equation. Let this stage be $f(t)$. Note that it 
 is defined for all $t$, 
 and hence $f$ is computable and satisfies the inequality at the start of the paragraph.
\end{proof}

\bigskip

Let $e_{f,g}(x)$ be equal to $|x|/(\log |x|)^5$ if 
for all $t$ either%\footnote{
    %In stead of $\ge 2$, the condition $\ge r$ can be used for any $r>1$
    %}
    \[ m_{f(t)}(t) \le \frac{6\log |x|}{|x|}
    \quad\quad \text{or} \quad\quad 
    \frac{m_{g(x,t)}(x)}{m_t(x)} \ge 2\,,
\]
otherwise let $e_{f,g}(x) = 1$.

\begin{proposition}\label{prop:e_is_sumtest}
  For some $c > 0$, some computable $f$ and for all computable $g$ the function $c \cdot e_{f,g}$ 
  is an upper semicomputable sumtest for~$m$.
\end{proposition}

\begin{proof}[Proof of Proposition~\ref{prop:e_is_sumtest}.]
 $e_{f,g}$ is upper semicomputable. Indeed, for at most finitely many $x$ the relation $|x|/(\log
 |x|)^5 \ge 1$ is false;
 for all other $x$, let the test be equal this value until we find a $t$ that
 does not satisfy the conditions.
 It remains to construct $f$ such that 
 \[
 \sum_x e_{f,g}(x)m(x) \le O(1)\,. 
 \]

 For some $x$ of length $n$, let $k = \log n - 3\log n - 3$, and suppose that $t_k$ and $f$ satisfy the
 conditions of Lemma~\ref{lem:t_k}. 
 One can calculate that $m_{f(t_k)}(t_k) \ge 2^{-k}/k^2 \ge (6\log n)/n$. 
 Thus the set of all $x$ of length $n$ such that $e_{f,g}(x) > 1$ has measure at most
 \[
   2^{-k + 1} \le O\left( \frac{(\log n)^3}{n}  \right)\,.
 \]
 We are now ready to bound the sum at the beginning of the proof. 
 It is sufficient to consider $x$ such that $e_{f,g}(x)>1$.
 \[
 \sum_x \left\{ e_{f,g}(x)m(x): e_{f,g}(x) > 1 \right\} 
 \le  \sum_n \frac{n}{(\log n)^5}\sum_{|x|=n} \left\{m(x): e_{f,g}(x)>1 \right\} \,,
 \]
 by the above bound on the measure of the sets in the right sum, this is at most
 \[
 \le O(1) \sum_n \frac{n}{(\log n)^5} \frac{(\log n)^3}{n} \le O(1)\sum_n \frac{1}{(\log n)^2} \le O(1)\,.
 \qedhere
 \]
\end{proof}

\begin{theorem}\label{th:largeTest}
  There exist an upper semicomputable sumtest for $m$ that exceeds 
  \[ 
    \Omega\left(|x|/(\log |x|)^5\right)
    \] 
  for some~$x$ of all large lengths.
\end{theorem}

\begin{proof}
  By  Proposition~\ref{prop:e_is_sumtest} it suffices to construct
  a computable $g$ and an $x$ of each large length such that $e_{f,g}(x) = |x|/(\log |x|)^5$. 
  Our construction works for any function~$f$.

  Fix a large~$n$.
  By construction of $e_{f,g}$, the function exceeds one on some $x$ of length $n$,  
  only if  $\log (1/m_t(x))$ increases each time $m_{f(t)}(t)$ is large. 
  Let $t_1, t_2, \dots$ be the set of all $t$ such that $m_{f(t)}(t) < 6(\log n)/n$. Clearly, 
  there are less than $n/(6\log n)$ such $t$ and the sequence can be enumerated uniformly in~$n$. 

  \medskip
  To construct $x$ we maintain a lexicographically ordered list that initially contains all strings
  of length~$n$.
  At stages $t_i$ we remove all $x$ from the list for which
  \[
    \log (1/m_{t_i}(x)) \le n - 6i\log n \,.
  \]
  %Because $t_i$ only exist for $i < n/(6 \log n)$, at each step the right hand 
  %is at least $6\log n$, thus for large $n$, the list does not become empty.
  %(here $\exp$ is taken with base two).
  Let $x$ be the first string in the list that is never removed.
  Such $x$ exist for large $n$, because in total there are less than $2^{n - 6\log n + 1}$ strings removed.
  \medskip

  By construction
  \[
    \log \left(1/m_{t_i}(x)\right) > n - 6i\log n\,. 
    \]
  Now we argue, why there is a computable $g$ such that 
  \begin{equation}\label{eq:gaolLarge}
    \log \left(1/m_{g(x,t_i)}(x)\right) \le n - 6i\log n - 2\log n + O(1)
  \end{equation}
  for all $i$ such that $t_i$ exists; and by construction of $e_{f,g}$ 
  this is sufficient for the theorem.
  After stage $t_i$, we remove at most 
  \[
  \exp (n - 6(i+1)\log n) + \exp(n - 6(i+2)\log n) + \ldots   < \exp (n + 1 - 6(i+1)\log n)
  \]
  strings.  Let $P(y|i,n)$ be the uniform distribution over the first $\exp (n + 1 - 6(i+1)\log n)$
  strings that remain in the lexicographically ordered list. 
  Note that we have $i < n/(6\log n)$, thus this list is never empty, and hence contains $x$.
  By universality there exist a $c$ such that for all strings $y$
  \[
   c \cdot m(y) \ge P(y|i,n)/(i^2n^2) \,.
  \]
  Now we construct the function $g(y,t)$. From $t$ and $n$ we can compute the largest $i$ such that 
  $t_i \le t$ if such $i$ exist (for $t < t_i$ the value of $g(y,t)$ may be anything).
  From $t_i$ and $i$ we can compute the lexicographically ordered list at stage $t_i$ 
  and evaluate $P(z|i,n)$ for all $z$ in the list. 
  Then we wait for the stage until the equation above is satisfied 
  for all $z$ in the list. Let this stage be $g(y,t)$.
  This implies 
  \[
    c \cdot m_{g(x,t_i)}(x) \ge P(x|i,n)/(i^2n^2) \ge \exp \left(-n - 1 + (6i+1)\log n\right)/n^4\,, 
  \]
  and this implies~\eqref{eq:gaolLarge}.
\end{proof}

The next corollary follows from the proof above.

\begin{corollary}\label{cor:nouniversal}
 There exist no test that is universal in the set of upper semicomputable sumtests for~$m$. 
\end{corollary}

\begin{proof}
 We show that for every test $e$, we can construct $f,g$ and infinitely many~$x$, such that 
 $e(x) \le O(1)$ and  $e_{f,g}(x) = |x|/(\log |x|)^5$.
 Remind the construction of $u_h$ and by Theorem~\ref{th:u_hierarchy} there exist 
 a computable $h$ be such that $e(x) \le O(u_h(x))$.

 Let $m_t(\cdot)$ be an approximation of $m(\cdot)$ from below such that 
 $m_1(x) \ge \Omega\left(2^{-|x|}\right)$. Now, for every $n$, 
 we follow the construction of $x$ from the proof above with the following
 modification: we do not start from a list of all $x$ of length $n$, 
 but from all $x$ of length $n$ such that
 \[ 
 m_{h(x,1)}(x|1) \le 2^{-n+1}\,.
  \] 
  There are less than $2^{n-1}$ strings with $m(x|1) > 2^{n-1}$, thus the list
  constains at least $2^n - 2^{n-1} = 2^{n-1}$ strings, and this is sufficient for the proof above.
 Let $t=1$ in the definition of $u_h$. This implies
 \[
  u_h(x) \le \frac{m_{h(x,1)}(x|1)}{m_1(x)} \le \frac{2^{-n+1}}{\Omega\left( 2^{-n} \right)} \le O(1) \,,
   \]
 thus $e(x) \le O(1)$. On the other hand, we can follow the proof above to construct $g$ and $x$ 
 such that $e_{f,g}(x) = n/(\log n)^5$.
 (The function $f$ is still obtained from  Lemma~\ref{lem:t_k}. 
 On the other hand, the function $g$ might be larger than in the proof above, and depends 
 on $h$. Indeed, it equals the stage on which $m(x)$
 increases sufficiently above $2^{-n+1}$, and this time is related to $h(x,1)$.)
\end{proof}

\bibliography{bib,eigen,statisticalCausalities,practCausalities,kolmogorov}

\begin{thebibliography}{1}

\bibitem{BauwensPhd}
B.~Bauwens.
\newblock {\em Computability in statistical hypotheses testing, and
  characterizations of independence and directed influences in time series
  using Kolmogorov complexity}.
\newblock PhD thesis, Ugent, May 2010.

\bibitem{msoph}
B.~Bauwens.
\newblock m-sophistication.
\newblock In {\em Proceedings of the 6th Conference on Computability in Europe
  (CiE)}, July 2010.

\bibitem{BauwensTerwijn}
B.~Bauwens and S.~Terwijn.
\newblock Notes on sum-tests and independence tests.
\newblock {\em Theory of Computing Systems}, 48:247--268, 2011.
\newblock open access.

\bibitem{limitComplexitiesRevisited}
L.~Bienvenu, A.~Muchnik, A.~Shen, and N.~Vereshchagin.
\newblock Limit com\-plexities revisited.
\newblock {\em Theory of Computing Systems}, 47(3):720--736, 2010.

\bibitem{Chaitin75}
G.J. Chaitin.
\newblock A theory of program size formally identical to information theory.
\newblock {\em J. Assoc. Comput. Mach.}, 22(3):329--340, 1975.

\bibitem{GacsNotes}
P.~G{\'a}cs.
\newblock Lecture notes on descriptional complexity and randomness.
\newblock \url{http://www.cs.bu.edu/faculty/gacs/papers/ait-notes.pdf},
  1988--2011.

\bibitem{LiVitanyi}
M.~Li and P.M.B. Vit{\'a}nyi.
\newblock {\em An Introduction to {K}olmogorov Complexity and Its
  Applications}.
\newblock Springer-Verlag, New York, 2008.

\bibitem{statCourse}
R.E. Walpole, R.H. Myers, S.L. Myers, and K.Ye.
\newblock {\em Probability \& Statistics for Engineers \& Scientists}.
\newblock Pearson Education, New Jersey, 2007.

\bibitem{Young}
G~A Young and R~L Smith.
\newblock {\em Essentials of statistical inference}.
\newblock Cambridge series in statistical and probabilistic mathematics.
  Cambridge Univ. Press, Cambridge, 2005.

\end{thebibliography}
\end{document}